%% file: ldfs-chordal-main-arxiv.tex
\pdfoutput=1 
\documentclass[a4paper,USenglish]{lipics-v2019}

\usepackage[utf8]{inputenc}

\usepackage{tikz}

\usepackage{style}

\tikzstyle kleinv=[circle,draw, minimum size=0.1cm]
\tikzstyle vertex=[circle, draw=black, fill=white, minimum size=15pt,inner 
sep=0pt]

\title{Linear Time LexDFS on Chordal Graphs} 

\author{Jesse Beisegel}{Brandenburg University of Technology, Cottbus, Germany} {}{}{}

\author{Ekkehard K\"ohler}{Brandenburg University of Technology, Cottbus, Germany}{}{}{}

\author{Robert Scheffler}{Brandenburg University of Technology, Cottbus, Germany} {}{}{}

\author{Martin Strehler}{Brandenburg University of Technology, Cottbus, Germany} {}{}{}

\authorrunning{J. Beisegel, E. K\"ohler, R. Scheffler and M. Strehler} 



\nolinenumbers 

\hideLIPIcs  

\EventEditors{John Q. Open and Joan R. Access}
\EventNoEds{2}
\EventLongTitle{42nd Conference on Very Important Topics (CVIT 2016)}
\EventShortTitle{CVIT 2016}
\EventAcronym{CVIT}
\EventYear{2016}
\EventDate{December 24--27, 2016}
\EventLocation{Little Whinging, United Kingdom}
\EventLogo{}
\SeriesVolume{42}
\ArticleNo{23}

\begin{document}
  
  	\maketitle

 	\begin{abstract}
	Lexicographic Depth First Search (LexDFS) is a special variant of a Depth First Search (DFS), which was introduced by Corneil and Krueger in 2008. While this search has been used in various applications, in contrast to other graph searches, no general linear time implementation is known to date. In 2014, Köhler and Mouatadid achieved linear running time to compute some special LexDFS orders for cocomparability graphs. In this paper, we present a linear time implementation of LexDFS for chordal graphs. Our algorithm is able to find any LexDFS order for this graph class. To the best of our knowledge this is the first unrestricted linear time implementation of LexDFS on a non-trivial graph class. In the algorithm we use a search tree computed by Lexicographic Breadth First Search (LexBFS).
	\end{abstract}

\section{Introduction}

\input{ldfs-chordal-intro.tex}

\section{Preliminaries}

\input{ldfs-chordal-prelim.tex}

\section{Search Orders and Trees of LexDFS}

\input{ldfs-chordal-ldfs}

 \section{LexDFS on Chordal Graphs}
 
 \input{ldfs-chordal-chordal}

 \section{Conclusion}

\input{ldfs-chordal-conclusion}

\bibliography{ldfs-chordal}
 
\end{document}

%% file: ldfs-chordal-intro.tex
Graph searches are among the most basic algorithms in computer science. Nevertheless, they are very powerful tools and can be used to compute many important graph properties. For example, \emph{Breadth First Search} (BFS) is the standard procedure for testing bipartiteness or computing shortest paths with respect to the number of edges.
Similarly, \emph{Depth First Search} (DFS) can be used in algorithms to find strongly connected components in directed graphs~\cite{tarjan1972depth} or to test for planarity~\cite{zbMATH03481829}. 

In 1976, Rose, Tarjan, and Lueker~\cite{rose1976algorithmic} proposed  a modified variant of BFS to compute perfect vertex elimination orders of chordal graphs. This search, since named \emph{Lexicographic Breadth First Search} (LexBFS), uses the ordering of the already visited vertices and visits the vertex with lexicographically largest neighborhood next. Habib et al.~\cite{habib2000lex} gave a simple linear time implementation of LexBFS using partition refinement, which, for example, also provides a linear time greedy algorithm for finding minimum colorings of chordal graphs.  

It was only in 2008 that a corresponding lexicographical variant for DFS was introduced by Corneil and Krueger~\cite{corneil2008unified}. Similar to LexBFS, this search computes perfect elimination orders on chordal graphs. Therefore, it can be used to find minimum colorings as well as all minimal separators and all maximal cliques on this graph class~\cite{xu2013moplex}. Besides this, LexDFS was used in the field of data mining to design an efficient hierarchical clustering algorithm~\cite{creusefond2017lexdfs}. However, no general linear time implementation of LexDFS is known to date. An implementation with running time in ${\cal{O}}(\min\{n^2,n+m\log n\})$ is given in~\cite{krueger2005graph}. Spinrad announced an ${\cal{O}}(m\log\log n)$-implementation~\cite{spinrad20??efficient} which has not been published as of yet. In~\cite{kohler2014linear}, Köhler and Mouatadid present the first linear time algorithm to compute a LexDFS cocomparability order, that is, a special class of LexDFS orders can be computed in linear time on cocomparability graphs using modular decomposition. However, there are LexDFS orders of cocomparability graphs that cannot be computed by this approach. Even more restricting, it is not possible to choose an arbitrary start vertex for the search. Nevertheless, this result can be used to design linear time algorithms which find minimum path covers~\cite{corneil2013ldfs}, maximum matchings~\cite{mertzios2018linear} as well as maximum independent sets, minimum clique covers and minimum vertex covers~\cite{corneil2016power} on cocomparability graphs.

An important concept in the theory of graph searches are search trees. Already in 1972, Tarjan~\cite{tarjan1972depth} gave a complete characterization of DFS-trees as so-called palm trees. However, no algorithm that determines whether a given spanning tree of a graph $G$ is a DFS-tree of $G$ was specified in that work. Using the concept of palm trees, Hopcroft and Tarjan developed a linear time algorithm for testing planarity of a graph~\cite{zbMATH03481829}. 
In 1985, Hagerup formulated the problem of checking whether a given spanning tree of $G$ can be obtained by a DFS and presented a linear time algorithm for this problem in~\cite{hagerup1985biconnected}. In the same year, Hagerup and Novak~\cite{hagerup1985recognition} presented a linear time algorithm for the recognition of BFS-trees. 
Similar results were obtained by Korach and Ostfeld~\cite{korach1989dfs} for DFS-trees and Manber~\cite{manber1990} for BFS-trees. Recently, Beisegel et al.~\cite{beisegel20??recognition,beisegel2019recognizing} studied the search tree recognition problem for LexBFS, LexDFS and other searches.

\paragraph*{Our Contribution.} In this paper, we give the first linear time implementation of LexDFS on chordal graphs. We show for all graphs that the computation of a LexDFS order is linear time equivalent to the construction of a LexDFS search tree, i.e., there are linear time reductions between both problems. The combination of this result with some properties of search trees of LexBFS on chordal graphs yields a linear time algorithm that can compute any LexDFS order of a given chordal graph. To the best of our knowledge this is the first unrestricted linear time implementation of LexDFS on a non-trivial graph class. Furthermore, we show that testing whether a given order is in fact a LexDFS order is linear time equivalent to the recognition of LexDFS search trees.

%% file: ldfs-chordal-prelim.tex
Throughout this paper, we consider finite, simple, undirected and connected graphs $G=(V,E)$ with $n=|V|$ vertices and $m=|E|$ edges. An edge between $u$ and $v$ is simply denoted by $uv$. For a vertex $v\in V$, the \emph{neighborhood} of $v$ is denoted by $N(v)$, i.e., $N(v)=\{u\in V~|~uv\in E\}$. For a subset $S\subseteq V$, we define the neighborhood as $N(S)=\{v\in V \setminus S ~|~ \exists u\in S:uv\in E\}$.

Given a subset $S$ of vertices in $G$, the \emph{subgraph of $G$ induced by $S$} is denoted by $G[S]$, where $V(G[S])=S$ and $E(G[S])=\{uv\in E(G)~|~u\in S,~v\in S\}$. The subgraph induced by $V(G)\setminus S$ is denoted by $G-S$ and, in the case that $S$ contains just one element, we simply write $G-v$ instead of $G-\{v\}$. 

A graph $G$ that contains no induced cycle of length larger than 3 is called \emph{chordal}. Other equivalent definitions of chordal graphs can be found in~\cite{brandstadt1999graph}. A \emph{tree} is an acyclic connected graph and a \emph{spanning tree} of a graph $G$ is an acyclic connected subgraph of $G$ which contains all vertices of $G$. A tree together with a distinguished \emph{root vertex} $s$ is said to be \emph{rooted}. In such a rooted tree a vertex $v$ is an \emph{ancestor} of vertex $w$ if $v$ is an element of the unique path from $w$ to the root $s$. In particular, if $ v $ is adjacent to $ w $, it is called the \emph{parent} of $ w $. A vertex $ w $ is called the \emph{descendant (child)} of $ v $ if $ v $ is the ancestor (parent) of $ w $.

A (connected) graph search is, in the most general sense, a mechanism for systematically visiting all vertices of a graph. Starting at a vertex $s\in V$, we expand the set of vertices $S$ beginning with $S=\{s\}$ by moving a vertex from $N(S)$ to $S$, which may also add new neighbors to $N(S)$ in consequence. The result of this procedure is a \emph{search order} $\sigma=(v_1=s,v_2,\dots,v_n)$ of the vertices of the graph listing the vertices in order of occurrence. For any linear vertex order $\sigma$ we write $u\prec_\sigma v$ if $u$ appears before $v$ in the order and say that $ u $ is \emph{to the left} of $ v $ and that $ v $ is \emph{to the right} of $ u $. Furthermore, $\sigma^-$ denotes the reverse order of $\sigma$, that is, $\sigma^-=(v_n,v_{n-1},\dots,v_1)$.

There are many graph search protocols which differ in the way in which a vertex from $N(S)$ is chosen next. The two most common graph searches are \emph{Breadth First Search} and \emph{Depth First Search} which can be simply described as using a queue and a stack to store the vertices in $N(S)$, respectively. Given a graph search protocol $\mathcal{P}$ and a vertex order $\sigma$, we say that $\sigma$ is a $\mathcal{P}$-\emph{order} if there exists a valid $\mathcal{P}$ search on $G$ that returns $\sigma$. 

In~\cite{corneil2008unified}, Corneil and Krueger present a characterizing \emph{four point property} of DFS orders.

\begin{lemma}[\cite{corneil2008unified}]\label{lemma:3point-dfs} A vertex order $\sigma$ is a DFS order of a graph $G = (V,E)$ if and only if for every triple $a \prec_\sigma b \prec_\sigma c$ where $ac \in E$ and $ab \notin E$ there is a vertex $d$ with $a \prec_\sigma d \prec_\sigma b$ such that $db \in E$.
\end{lemma}

In the same paper, the authors introduced \emph{Lexicographic Depth First Search} (LexDFS, see Algorithm~\ref{alg:lexdfs}), a variant of DFS which uses labels and their lexicographic ordering to break ties during the search. 

\begin{algorithm2e}
\caption{Lexicographic Depth First Search}
	\label{alg:lexdfs}
	\KwIn{Connected graph $G=(V,E)$ and a distinguished vertex $ s \in V $}
	\KwOut{Ordering $ \sigma $ of $V$ starting at $s$}
	\Begin{
		$ label(s) \leftarrow (0) $\;
		
		\lForEach{vertex $v \in V-{s}$}{assign to $ v $ the empty label}
		
		\For{$ i \leftarrow 1 $ to $ n $}{
			pick an unnumbered vertex $ v $ with lexicographically largest label\;\label{alg:line:labelchoose:ldfs}
			$ \sigma(i) \leftarrow v$\;
			\lForEach{unnumbered vertex $ w \in N(v) $}{prepend $ i $ to $ label(w) $}}
	}
\end{algorithm2e}

The idea of using such a lexicographic ordering of labels originates from an algorithm for the calculation of perfect elimination orders of chordal graphs, given by Rose, Lueker, and Tarjan~\cite{rose1976algorithmic}, since named \emph{Lexicographic Breadth First Search} (LexBFS, see Algorithm~\ref{alg:lexbfs}).

\begin{algorithm2e}
\caption{Lexicographic Breadth First Search}
	\label{alg:lexbfs}
	\KwIn{Connected graph $G=(V,E)$ and a distinguished vertex $ s \in V $}
	\KwOut{Ordering $ \sigma $ of $V$ starting at $s$}
	\Begin{
		$ label(s) \leftarrow (n) $\;
		
		\lForEach{vertex $v \in V-{s}$}{assign to $ v $ the empty label}
		
		\For{$ i \leftarrow 1 $ to $ n $}{
			pick an unnumbered vertex $ v $ with lexicographically largest label\;\label{alg:line:labelchoose:lbfs}	
			$ \sigma(i) \leftarrow v$\;
			\lForEach{unnumbered vertex $ w \in N(v) $}{append $ (n-i) $ to $ label(w) $}}\label{alg:line:labelupdate}		
	}
\end{algorithm2e}

Both LexDFS and LexBFS are special variants of the standard searches and, thus, every LexDFS order is also a DFS order and every LexBFS order is also a BFS order. In both algorithms, the vertices are labeled by their already visited neighbors (see line~\ref{alg:line:labelupdate} in Algorithm~\ref{alg:lexdfs} and~\ref{alg:lexbfs}). While in LexDFS vertices visited later in the search have a larger significance for the lexicographic order of the label, in LexBFS it the opposite, i.e., vertices visited earlier have a larger impact.

Corneil and Krueger~\cite{corneil2008unified} also present a four point property of LexDFS orders.

\begin{lemma}[\cite{corneil2008unified}]\label{lemma:3point-lexdfs} A vertex order $\sigma$ is a LexDFS order of a graph $G = (V,E)$ if and only if for every triple $a \prec_\sigma b \prec_\sigma c$ where $ac \in E$ and $ab \notin E$ there is a vertex $d$ with $a \prec_\sigma d \prec_\sigma b$ such that $db \in E$ and $dc \notin E$.
\end{lemma}

A variant of LexDFS and LexBFS is the technique of ``multisweeping''. This describes the multiple application of some graph search, where each run of the search uses the order given by the previous application as a so-called \emph{``tie-break'' rule}, that is, a priority list which decides which vertex can be visited next in those cases where the given search paradigm allows several different options. It was first used by Simon~\cite{simon1991new} in an algorithm for the recognition of interval graphs which is flawed as was shown by Ma~\cite{ma????}. Nevertheless, ``multisweeping'' has proven to be very fruitful in recent years~\cite{corneil2013ldfs,zbMATH05822744}. In the case of LexDFS, this technique implies a new search scheme known as LexDFS$ ^+ $: Given a linear order $ \rho $ of the vertices, LexDFS$ ^+ (\rho)$ is computed by executing a regular LexDFS with the modification that in line~\ref{alg:line:labelchoose:ldfs} of Algorithm~\ref{alg:lexdfs} the rightmost element with regard to $ \rho $ is chosen among all vertices with lexicographically largest label. The searches DFS$^+$, BFS$^+$ and LexBFS$^+$ are defined analogously. Note that all these searches yield unique orders, as there are no more ties to break in the algorithms.

It is not difficult to see that, given the reverse of a search order as tie break, such a procedure yields that same order again.

\begin{observation}\label{obs:lexdfs+}
Let $G = (V,E)$ be a graph and let $\sigma$ be a vertex order of $G$. The order $\sigma$ is a LexDFS order of $G$ if and only if LexDFS$^+(\sigma^-)$ is equal to $\sigma$. This also holds for LexBFS and LexBFS$^+(\sigma^-)$.
\end{observation}

Using a technique called \emph{partition refinement}, LexBFS can be implemented in linear time~\cite{habib2000lex}. Given a set $S$, we call $ \mathcal{Q} =(Q_1, Q_2, \ldots , Q_k)$ a partition of $ S $ if $S=\bigcup_{i_=1}^{k} Q_i$ with pairwise disjoint sets $Q_i \cap Q_j=\emptyset$, $i\ne j$. Note that a partition is an \emph{ordered} list of subsets. We say that a subset $S'\subseteq S$ refines $\mathcal{Q}$ if $Q_i$ is replaced by a subpartition $(A_i,B_i)$ where $A_i=Q_i \cap S'$ and $B_i=Q_i\setminus A_i$ whenever both sets are non-empty.  In particular, for LexBFS we start with $\mathcal{Q}=(V)$ and starting vertex $s$. Now, we refine $\mathcal{Q}$ with $\{s\}$ which separates $s$ in a single set. Afterwards, we refine with $N(s)$. In the first iteration, this yields the partition $\left(\{s\},N(s),V\setminus( N(s)\cup \{s\})\right)$. The vertex whose neighborhood is used to refine the partition classes is called a \emph{pivot}. Choose the next pivot $v$ from $N(s)$ and refine with $\{v\}$ and then $N(v)$. For $Q_i=\{v\}$, repeat refining using an element from the set $Q_{i+1}$ as the next pivot, maintaining the order of the partition classes created so far. As shown in~\cite{habib2000lex}, this final partition can be computed in linear time and it is actually a LexBFS order. The pivot is just the vertex visited by the search and it \emph{pulls} its neighbors to the front of each set. Unfortunately, no linear time implementation of partition refinement for LexDFS is known to date.

Usually, a graph search is associated with a search tree which is a spanning tree of the graph. Given a BFS order $\sigma=(v_1,\dots,v_n)$, a vertex $v_i$ is typically connected to the leftmost neighbor in $(v_1,\dots,v_{i-1})$, i.e., it is connected to the vertex that was current at the point at which $v_i$ was added to $N(S)$. Contrary, given a DFS order $\sigma=(v_1,\dots,v_n)$, a vertex $v_i$ is connected to the rightmost neighbor in $(v_1,\dots,v_{i-1})$, i.e., it is connected to the neighbor which occurred last before $v_i$ itself was visited. These two different approaches of constructing a search tree  give rise to the following definition.

\begin{definition}[\cite{beisegel2019recognizing}]\label{def:tree}
	Given a search order $ \sigma=(v_1, \ldots , v_n) $ of a given search on a connected graph $ G=(V,E) $, the \emph{first-in tree} (or \cf-tree) of $\sigma$ is the tree consisting of the vertex set $ V $ and an edge from each vertex to its leftmost neighbor in~$ \sigma $. 
	The \emph{last-in tree} (or \cl-tree) of $\sigma$ is the tree consisting of the vertex set $ V $ and an edge from each vertex $ v_i $ to its rightmost neighbor $ v_j $ in $ \sigma $ with $ j < i $.
	In both cases, $v_1$ is the \emph{root} of the search tree.
\end{definition}

The notation of \cf-trees and \cl-trees was introduced in~\cite{beisegel2019recognizing}, where the recognition problem of these search trees was studied. In contrast to the original definition, we always assume that a search tree has a designated root. In particular, two search trees on a graph $G$ are equal if they use the same edge set and if they have the same root. Given a search protocol $\mathcal{P}$ and a spanning tree $T$ of $G$ rooted in $s$, we say that $T$ is an \cl-\emph{tree} (\cf-\emph{tree}) of $\mathcal{P}$ on $G$ if there is a $\mathcal{P}$-order of $G$ starting at $s$ with \cl-tree (\cf-tree) $T$.

Although, it would be most natural to consider the \cf-tree for LexBFS, the \cl-tree of LexBFS is a key ingredient in our procedure on chordal graphs.

%% file: ldfs-chordal-ldfs.tex
Given a search order, it is easy to construct the corresponding search tree in linear time by simply using Definition~\ref{def:tree}. However, if we are only given a search tree, then it is not immediately clear how to find a search order that results in this tree. In this section, we present a linear time algorithm that computes a LexDFS order for a given \cl-tree of LexDFS. This shows that both recognition and creation of search orders and \cl-trees is linear time equivalent in the case of LexDFS.

The main idea of this algorithm is to use a special tie-break rule in form of a linear order of the vertices $ \tau $ such that a simple run of DFS$ ^+(\tau) $ on the tree is a LexDFS order of $ G $ with tree $ T $. This tie-break rule $ \tau $ is computed by using a form of partition refinement which moves from the leaves of the tree towards its root. The pseudo code of this procedure is given in Algorithm~\ref{algo:ordering}.

\begin{algorithm2e}
	\caption{\textsc{Ordering}($G$, $T$, $s$, $\rho$)}\label{algo:ordering}	
	\KwIn{Connected graph $G=(V,E)$, an \cl-tree $T$ of DFS on $G$ rooted in $s \in V$, order $\rho$ of $V$ ending with vertex $s$}
	\KwOut{order $\sigma$ of $V$ starting at $s$}
	\Begin{
		$\beta \leftarrow$  reverse of a BFS order of $ T $ starting at $ s $\;\label{line:begin1st}
		$\mathcal{Q} \leftarrow (V)$\;
		\For{$i$ $\leftarrow$ 1 to $n$}{\label{line:for1}
			$v \leftarrow \beta(i)$\;
			refine $\mathcal{Q}$ with $\{w \in N(v) ~|~ w \prec_\beta v\}$\;
		}\label{line:for2}
		order every set in $\mathcal{Q}$ with respect to $\rho^-$ and move $ \{s\} $ to the leftmost position\;
		$\tau \leftarrow$ reverse of the final order of vertices in $\mathcal{Q}$\;\label{line:end1st}
		$ \sigma \leftarrow $ DFS$^+ (\tau) $ on $ G $\;			
		\Return{$\sigma$}\;
	}	
\end{algorithm2e}

Before we begin with the analysis of this algorithm we present some general results on \cl-trees of DFS. The first is a lemma by Tarjan~\cite{tarjan1972depth} which characterizes \cl-trees of DFS.  

\begin{lemma}[\cite{tarjan1972depth}]\label{lemma:dfstrees}
	Let $ G=(V,E) $ be a graph and let $ T $ be a spanning tree of $ G $. Then $ T $ is an \cl-tree of $ G $ generated by DFS if and only if for each edge $ uv \in E $ it holds that either $ u $ is an ancestor of $ v $ in $ T $ or $ v $ is an ancestor of $ u $ in $ T $.
\end{lemma}

In order to make sure that Algorithm~\ref{algo:ordering} returns a DFS order of $ G $ with \cl-tree $ T $, we prove the following statement.

\begin{lemma}\label{lemma:tree}
Let $T$ be an $\cl$-tree of some DFS on $G$ rooted in $s$ and let $\sigma$ be a DFS order of $ T $ starting at $ s $. Then $\sigma$ is a DFS order of $G$ with $\cl$-tree $T$.
\end{lemma}
\begin{proof}
We show that $\sigma$ is a DFS order of $G$ by proving that it fulfills the characterization given in Lemma~\ref{lemma:3point-dfs}. Let $a$, $b$ and $c$ be three vertices in $G$ with $a \prec_\sigma b \prec_\sigma c$, $ac \in E(G)$ and $ab \notin E(G)$. We have to show that there is a vertex $ d $ with $ a \prec_\sigma d \prec_\sigma b $ such that $ db \in E(G) $. Assume that $ac \in E(T)$. As $\sigma$ is a DFS order of the tree $T$, there is a vertex $d$ with $a \prec_\sigma d \prec_\sigma b$ and $db \in E(T) \subseteq E(G)$ due to Lemma~\ref{lemma:3point-dfs}. Therefore, we can assume that $ac$ is not contained in $E(T)$. By Lemma~\ref{lemma:dfstrees} this implies that $a$ is an ancestor of $c$ in $T$. Let $P = (a = w_1, \ldots, w_k = c)$ be the unique path between $a$ and $c$ in $T$. As $\sigma$ is a DFS order of $T$, it holds that $a \prec_\sigma w_2 \prec_\sigma \ldots \prec_\sigma c$. If there is an $1 < i < k$ such that $w_i = b$, then $w_{i-1}$ is a vertex between $a$ and $b$ in $\sigma$ with $w_{i-1}b \in E(G)$. Otherwise, there exists an $1 \leq i < k$ with $w_i \prec_\sigma b \prec_\sigma w_{i+1}$ and $w_iw_{i+1} \in E(T)$. As in the first case there exists a vertex $d$ with $a \preceq_\sigma w_i \prec_\sigma d \prec_\sigma b$ and $db \in E(T) \subseteq E(G)$. By Lemma~\ref{lemma:3point-dfs}, this proves that $\sigma$ is a DFS order.

Let $T'$ be the \cl-tree of $\sigma$ with regard to $ G $ and assume for contradiction that there is an edge $uv$ in $T'$ that is not part of $T$. Due to Lemma~\ref{lemma:dfstrees}, we can assume without loss of generality that $u$ is an ancestor of $v$ in $T$. Since $uv$ is not part of $T$, vertex $u$ is not the parent of $v$ in $T$. Let $w$ be the parent of $v$ in $T$. Note that this means that $w$ is a descendant of $u$ in $T$. Since $ \sigma $ is a DFS order on $ T $, vertex $w$ must be to the left of $v$ and to the right of $u$ in $\sigma$. Since $vw \in E(G)$, edge $uv$ cannot be part of $T'$, as $T'$ is the $\cl$-tree of $\sigma$; a contradiction.
\end{proof}

With these results on DFS we can proceed to the analysis of Algorithm~\ref{algo:ordering}. First we will prove correctness.

\begin{theorem}\label{thm:tree-ordering}
 Let $T$ be an $\cl$-tree of some DFS on $G$ rooted in $s$ and let $ \rho $ be an arbitrary order of $ V $ ending in $ s $. Let $\sigma$ be the order produced by Algorithm~\ref{algo:ordering} with input $(G,T,s,\rho)$. Then $T$ is an $\cl$-tree of LexDFS rooted in $s$ if and only if $\sigma$ is a LexDFS order of $G$.
\end{theorem}
 \begin{proof}
  Assume $\sigma$ is a LexDFS order of $G$. Due to Lemma~\ref{lemma:tree}, the $\cl$-tree of $\sigma$ is $T$ and, therefore, it is an $\cl$-tree of LexDFS.

  For the other direction, assume that $T$ is an $\cl$-tree of LexDFS. Let $\sigma^*$ be a LexDFS order of $G$ such that the $\cl$-tree of $\sigma^*$ is $T$ and the common prefix of $\sigma$ and $\sigma^*$ is maximal among all LexDFS orders with $\cl$-tree $T$. If $ \sigma $ and $ \sigma^* $ are equal, then we are done. Otherwise let $i \in \{1,\ldots, n\}$ be the first index for which $v = \sigma(i) \neq \sigma^*(i) = v^*$. By Lemma~\ref{lemma:tree}, both $\sigma$ and $\sigma^*$ are DFS orders with $\cl$-tree $T$ and $v$ and $v^*$ have the same parent $ p $ in $T$. If $v$ and $v^*$ have the same neighborhood in the set $S = \{\sigma(j) ~|~ j < i\}$, then $v$ could have been taken by LexDFS instead of $v^*$ and this choice would not have had an impact on the $\cl$-tree of the order, due to Lemma~\ref{lemma:dfstrees}. Hence, there must be a vertex $w \in S$ with $wv^* \in E(G)$ and $wv \notin E(G)$ and for all vertices $x$ with $w \prec_\sigma x \prec_\sigma v$ it holds that both $v$ and $v^*$ are adjacent to $x$ or both are not adjacent to $x$. Note that $w$ is an ancestor of both $v$ and $v^*$ in $T$ and therefore, it is to the right of both vertices in $\beta$.
  
  However, this means that before the iteration of the \texttt{for}-loop in lines \ref{line:for1}--\ref{line:for2}, where we consider vertex $w$, both $v$ and $v^*$ are in the same set of $\cal Q$. After this iteration, vertex $v^*$ is in a set of $ \cal Q $ to the left of the set containing $v$. Therefore, $v^*$ is to the right of $v$ in $\tau$, as $ \tau $ uses the reverse order of $ \cal Q $. Thus, the search DFS$ ^+(\tau) $ visits $ v^* $ before $ v $, as both are children of $ p $; a contradiction to $ v $ being to the left of $ v^* $ in $ \sigma $.
 \end{proof}
 
As seen in the proof, it is not necessary to use the reverse of a BFS order for $\beta$. Any order will suffice, where for every vertex $w$ all ancestors in $T$ are to the right of $w$ in the order. Having shown that Algorithm~\ref{algo:ordering} returns a correct LexDFS order for any \cl-tree of LexDFS, we will now evaluate its running time.

\begin{lemma}\label{lemma:running-time-ordering}

 Algorithm~\ref{algo:ordering} has running time in $\O(n+m)$.
\end{lemma}
\begin{proof}
 Algorithm~\ref{algo:ordering} begins with an execution of BFS which can be done in linear time. In the \texttt{for}-loop
 we iterate through the neighborhood of every vertex exactly once. Thus, the overall costs are in $\O(n+m)$. To sort the sets of $\cal Q$ with respect to $\rho^-$ we iterate through $\rho^-$ and move the considered vertex to the end of its set. The final DFS$^+(\tau)$ can be executed in linear time by first sorting the neighborhoods of all vertices with respect to $\tau$. 
\end{proof}

The last results imply that the construction of an \cl-tree of LexDFS is linear time equivalent to the computation of a LexDFS order.

\begin{theorem}
 For a given graph family $ \cal G $ and $\O(M) \supseteq \O(n+m)$ the following two statements are equivalent:
 \begin{enumerate}
 	\item There is an algorithm with running time in $\O(M)$ that computes a LexDFS order for any graph in $ \cal G $ and any starting vertex $s$.
 	\item There is an algorithm with running time in $\O(M)$ that creates an $\cl$-tree of LexDFS for any graph in $\cal G$ and any root $s$.
 \end{enumerate}
\end{theorem}
\begin{proof}
 It is easy to see that the $\cl$-tree of an arbitrary vertex order can be constructed in $\O(n+m)$. Therefore, an $\O(M)$-algorithm for the computation of a LexDFS order starting at $s$ directly implies an $\O(M)$-algorithm for the creation of an $\cl$-tree of LexDFS rooted in $s$.
 
 If, on the other hand, we can compute an $\cl$-tree of LexDFS rooted in $s$ in time $\O(M)$, then we can use Algorithm~\ref{algo:ordering} to create a corresponding LexDFS order in linear time, due to Theorem~\ref{thm:tree-ordering} and Lemma~\ref{lemma:running-time-ordering}.
\end{proof}

This linear time equivalence does not only hold for the computation but also for the recognition of $\cl$-trees and orders of LexDFS. To prove this we need the following two technical lemmas.

\begin{lemma}\label{lemma:rho}
 	Let $T$ be an $\cl$-tree of a DFS on $G$ rooted in $s$, let $\rho$ be an arbitrary order of $V$ ending with $s$ and let $\sigma$ be the order produced by Algorithm~\ref{algo:ordering} with input $(G,T,s,\rho)$. Furthermore, let $v$ and $w$ be two vertices in $G$ with $ v \prec_\sigma w $ which have the same parent in $T$ and the same neighborhood in the set $Y=\{x~|~x \prec_\sigma v\}$. Then $v$ is to the right of $w$ in $\rho$.
\end{lemma}
\begin{proof}
 	As $v$ is to the left of $w$ in $ \sigma $ by assumption, vertex $ v $ was taken before $ w $ in $ DFS^+(\tau) $. Since $ v $ and $ w $ have the same parent in $ T $, it holds that $ v $ is to the right of $ w $ in $ \tau $. If the vertex $v$ is pulled by a vertex $x$ in the \texttt{for}-loop of Algorithm~\ref{algo:ordering}, then $x$ is adjacent to $v$ and has a smaller distance to the root $s$ in $T$. By Lemma~\ref{lemma:dfstrees}, vertex $x$ is an ancestor of $v$ in $ T $ and $ x $ is to the left of $v$ in $\sigma$. As $ v $ and $ w $ have the same neighborhood in $ Y $, the vertex $x$ is also adjacent to $w$ and pulls it, too. This implies that $v$ and $w$ are in the same set of $\cal Q$ after the \texttt{for}-loop. Therefore, $ v $ has to be to the right of $ w $ in $ \rho $, as it is to the right of $ w $ in $ \tau $.
\end{proof}

\begin{lemma}\label{lemma:sorted}
  Let $T$ be an $\cl$-tree of LexDFS on $G$ rooted in $s$ and let $\sigma$ be a vertex order of a graph $G$ starting at $s$ whose corresponding $\cl$-tree is $T$. Algorithm \ref{algo:ordering} returns $\sigma$ for input $(G,T,s,\sigma^-)$ if and only if $\sigma$ is a LexDFS order of $G$.
\end{lemma}
\begin{proof}
 If Algorithm~\ref{algo:ordering} returns $\sigma$ for input $(G,T,s,\sigma^-)$, then, by Theorem~\ref{thm:tree-ordering}, $\sigma$ is a LexDFS order of $G$. 
 
 Therefore, we assume that $\sigma$ is a LexDFS order and Algorithm~\ref{algo:ordering} returns the LexDFS order $\sigma^*$ for input $(G,T,s,\sigma^-)$ with $\sigma \neq \sigma^*$. Let $i \in \{1,\ldots,n\}$ be the first index where $v = \sigma(i) \neq \sigma^*(i) = v^*$ and let $\sigma_i$ be the prefix of the first $i-1$ elements of $\sigma$ (and $\sigma^*$). It follows that $v$ and $v^*$ have the same neighborhood in $\sigma_i$ and, thus, the same parent in $T$. Since $v$ is to the right of $v^*$ in $\sigma^-$ it follows from Lemma~\ref{lemma:rho} that $v$ must be to the left of $v^*$ in $\sigma^*$; a contradiction.
\end{proof}

Now we can prove the linear time equivalence of tree recognition and order verification for LexDFS.

\begin{theorem}\label{theo:recognition}
  For a given graph family $ \cal G $ and $\O(M) \supseteq \O(n+m)$ the following two statements are equivalent:
  \begin{enumerate}
 	\item There is an algorithm with running time in $\O(M)$ that checks for any graph $G$ in $ \cal G $ and any vertex $s$ whether a given order beginning in $ s $ is a LexDFS order of $G$.
 	\item There is an algorithm with running time in $\O(M)$ for any graph $G$ in $\cal G$ and any vertex $s$ that checks whether a given spanning tree rooted in $ s $ is an $\cl$-tree of LexDFS on $G$.
  \end{enumerate}
\end{theorem}
\begin{proof}
	Assume we have an algorithm $\cal A$ for the recognition of LexDFS orders with running time in $\O(M)$. For a given spanning tree $T$ of $G$ rooted in $s$ we first decide in linear time whether $T$ is an $\cl$-tree of DFS (see~\cite{hagerup1985biconnected,korach1989dfs}). If not, then it is not an $\cl$-tree of LexDFS. Otherwise, we execute Algorithm~\ref{algo:ordering} with the input $(G,T,s,\rho)$, where $ \rho $ is an arbitrary order of the vertices of $ G $, and get the vertex order $\sigma$ as result in  linear time. Due to Theorem~\ref{thm:tree-ordering}, $T$ is an $\cl$-tree of LexDFS if and only if $ \sigma $ is an LexDFS order of $ G $. We use $\cal A$ to decide this in time $\O(M)$.
	
	Now assume we have an algorithm $\cal A$ for the recognition of $\cl$-trees of LexDFS with running time in $\O(M)$ and get a vertex order $\sigma$ starting at $s$. We first create the $\cl$-tree $T$ of $\sigma$ in linear time and check whether $T$ is an $\cl$-tree of LexDFS in time $\O(M)$. If not, $\sigma$ is not an LexDFS order of $G$. Otherwise, we call Algorithm~\ref{algo:ordering} with input $(G,T,s,\sigma^-)$. Due to Lemma~\ref{lemma:sorted}, the resulting vertex order  is equal to $\sigma$ if and only if $\sigma$ is a LexDFS order of $G$.
\end{proof}
 
 Note that this result does not hold for search orders and their corresponding trees in general (if $\P \neq \NP$). Beisegel et al.~\cite{beisegel20??recognition,beisegel2019recognizing} show for example that the recognition problem of $\cf$-trees of both LexBFS and LexDFS is $\NP$-complete, whereas it is easy to recognize the corresponding orders.

%% file: ldfs-chordal-chordal.tex

We will now use the results of the last section to derive a linear time implementation of LexDFS for chordal graphs. We first show that LexBFS and LexDFS have the same set of \cl-trees on chordal graphs. This fact is also implied by a more general result in~\cite{beisegel20??recognition}. Since this work has not been published yet and we only need a special case here, we give an alternative proof for the sake of completeness.

In~\cite{berry2009maximal}, Berry et al.~show that a whole range of different graph search schemes share the same set of search orders on chordal graphs. Among these searches are variants of both LexDFS and LexBFS, called \emph{CompLexDFS} and \emph{CompLexBFS}, respectively. For these algorithms we replace line~\ref{alg:line:labelchoose:ldfs} in both Algorithm~\ref{alg:lexdfs} and~\ref{alg:lexbfs} by ``choose a component $C$ of the graph induced by the unnumbered vertices and take a vertex in $C$ with lexicographically largest label''.

\begin{lemma}[\cite{berry2009maximal}]\label{lemma:berry}
	For any chordal graph $ G $ a linear vertex order is a CompLexDFS order if and only if it is a CompLexBFS order.
\end{lemma}

We now show that both LexDFS and LexBFS compute the same $\cl$-trees as their respective Comp-variants for any graph.

\begin{lemma}\label{lemma:comp}
A spanning tree $T$ of a graph $G$ rooted in $s$ is an $\cl$-tree of LexDFS (LexBFS) on $G$ if and only if $T$ is an $\cl$-tree of CompLexDFS (CompLexBFS) on $G$.
\end{lemma}
\begin{proof}
Since every order of LexDFS is also an order of CompLexDFS, every $\cl$-tree of LexDFS on $G$ is also an $\cl$-tree of CompLexDFS.

For the reverse we first introduce some technical definitions. Let $ \tau =(w_1, \ldots , w_n) $ be some order of the vertices of $ G $. We define 
$ C_\tau(w_i) $ to be the connected component of $ G - \{w_1, \ldots, w_{i-1}\} $ containing $ w_i $. Now, consider a CompLexDFS order $\sigma$ of $G$ with $\cl$-tree $T$.  Let $\sigma^*$ be the LexDFS$^+(\sigma^-)$ order of $G$. We claim that $T$ is the $\cl$-tree of $\sigma^*$. To this end, we show that $ C_{\sigma}(v)=C_{\sigma^*}(v) $ for every vertex $v \in V$. Furthermore, we show that for every vertex $ w \in C_{\sigma}(v)=C_{\sigma^*}(v) $ it holds that the label of $ w $ at point where $v$ is chosen in $ \sigma $ is the same as the label of $ w $ when $ v $ is chosen in $ \sigma^* $.

Assume for contradiction that $v$ is the leftmost vertex in $\sigma$ which does not fulfill both of these properties. Let $w$ be the rightmost vertex in $\sigma$ with $w \prec_\sigma v$ such that $w$ has a neighbor in $C_{\sigma}(v)$. Due to choice of $v$, it holds that $ C_{\sigma}(w) = C_{\sigma^*}(w) $ and both components are labeled the same at the moment $ w $ is chosen in the respective search. As the labels cannot be changed from outside of the component, there must be a vertex in $C_{\sigma}(v)$ that is between $w$ and $v$ in $\sigma^*$. Let $x$ be the leftmost vertex in $\sigma^*$ with this property. At the point where $x$ is chosen by $ \sigma^*$ the labels of both $x$ and $v$ are the same as in $\sigma$ at the point when $v$ was chosen. Therefore, the labels of $x$ and $v$ must be the same at the point where $ x $ was chosen in $ \sigma^* $. However, $v$ is to the right of $x$ in $\sigma^-$ and it has to be chosen before $x$ in $ \sigma^* $; a contradiction. 

Now, let $T^*$ be the $\cl$-tree of $\sigma^*$ and assume that the parent of vertex $y$ in $T$ is $p$ and the parent of $y$ in $T^*$ is $p^* \neq p$. Due to the observation above, both $p$ and $p^*$ must be to the left of $y$ in both $\sigma$ and $\sigma^*$. Therefore, it holds that $p \prec_{\sigma^*} p^*$ and $p^* \prec_\sigma p$. However, this is a contradiction to the observation above since $p$ would be in $C_\sigma(p^*)$ but not in $C_{\sigma^*}(p^*)$.

Since no special property of LexDFS and CompLexDFS is used in the proof above, the claim also holds for LexBFS and CompLexBFS.
\end{proof}

Combining Lemmas~\ref{lemma:berry} and~\ref{lemma:comp} yields the following corollary.

\begin{corollary}\label{corol:ldfs-lbfs-trees}
Let $G=(V,E)$ be a chordal graph and $T$ be a spanning tree of $G$ rooted in $s \in V$. The tree $T$ is an $\cl$-tree of LexDFS of $G$ if and only if $T$ is an $\cl$-tree of LexBFS of $G$.
\end{corollary}

Using this corollary, we can compute an \cl-tree of LexDFS rooted in vertex $ s $ for any chordal graph $ G $ by using LexBFS. This tree can then be used as the input for Algorithm~\ref{algo:ordering} to return a LexDFS order for $ G $. Furthermore, it is possible to implement LexDFS$ ^+ $ in linear time for chordal graphs using the same approach (see Algorithm~\ref{algo:ldfs_ordering}).

\begin{algorithm2e}
	\caption{LexDFS$ ^+ $ on chordal graphs}\label{algo:ldfs_ordering}	
	\KwIn{Chordal graph $G=(V,E)$, vertex $s \in V$, order $\rho$ of $V$ ending with $s$}
	\KwOut{LexDFS$^+(\rho)$ order $\sigma$ of $G$ starting at $s$}
	\Begin{
		$\pi \leftarrow$ LexBFS$^+(\rho)$ order of $G$\;
		$T \leftarrow$ \cl-tree of $\pi$\;
		$\sigma \leftarrow$ \textsc{Ordering}($G$, $T$, $s$, $\rho$)\;
		\Return{$\sigma$}\;
	}	
\end{algorithm2e}

 \begin{theorem}\label{theo:ldfs}
 Let $G = (V,E)$ be a chordal graph, $s$ be a vertex in $V$ and $\rho$ be an arbitrary order of $V$ ending in $s$. Then for input $(G,s,\rho)$ Algorithm~\ref{algo:ldfs_ordering} produces a LexDFS$^+(\rho)$ order of $G$ starting at $s$ in time $\O(n+m)$.
 \end{theorem}
 \begin{proof}
 Due to Corollary~\ref{corol:ldfs-lbfs-trees}, the tree $T$ is an $\cl$-tree of LexDFS. By Theorem~\ref{thm:tree-ordering}, \textsc{Ordering}($G$, $T$, $s$, $\rho$) produces a LexDFS order of $G$ starting at $s$.
 
 It remains to show that $\sigma$ is also the LexDFS$^+(\rho)$ order. Let $\sigma^*$ be the LexDFS$^+(\rho)$ order of $G$ and assume that $\sigma \neq \sigma^*$. Let $i \in \{1,\ldots,n\}$ be the first index where $v = \sigma(i) \neq \sigma^*(i) = v^*$ and let $\sigma_i$ be the prefix of the first $i-1$ elements of $\sigma$ (and $\sigma^*$). It follows that $v$ and $v^*$ have the same neighborhood in $\sigma_i$ and $v^*$ must be to the right of $v$ in $\rho$.
  
 Assume that $v$ is to the left of $v^*$ in the LexBFS$^+(\rho)$ order $\pi$. Since $v \prec_\rho v^*$, vertex $v$ had a larger label than $v^*$ at the point where it was chosen in $\pi$. This implies that there is a vertex $w$ with $w \prec_\pi v \prec_\pi v^*$ such that $wv \in E(G)$ but $wv^* \notin E(G)$. Due to Lemma~\ref{lemma:dfstrees}, vertex $w$ has to be an ancestor of $v$ in $T$ and, therefore, $w$ is in $\sigma_i$. This is a contradiction as $v$ and $v^*$ have the same neighbors in $\sigma_i$.
 
 Therefore, we can assume that $v^*$ is to the left of $v$ in $\pi$. If $v$ and $v^*$ have the same parent in $T$, then vertex $v$ has to be to the right of $v^*$ in $\rho$, due to Lemma~\ref{lemma:rho}; a contradiction.
 Thus, assume that $p$ is the parent of $v$ but not the parent of $v^*$ in $T$. However since $p$ is in $\sigma_i$, vertex $v^*$ is adjacent to $p$ in $G$ and, therefore, is a descendant of $p$ in $T$, due to Lemma~\ref{lemma:dfstrees}. Let $x$ be the unique child of $p$ in $T$, which is an ancestor of $v^*$. Note that $x \neq v$ since otherwise $v \prec_\pi v^*$. Furthermore, it holds that both $x \prec_\pi v$ and $v \prec_\sigma x$ and every neighbor of $x$, which is to the left of $x$ in $\pi$, is an element of $\sigma_i$, due to the choice of $i$. 
 If $x$ has the same neighborhood in $\sigma_i$ as $v$ and $v^*$, then $v^* \prec_\rho x$ and, due to Lemma~\ref{lemma:rho}, it holds that $x \prec_\sigma v$; a contradiction. Thus, $x$ has a neighbor $y$ in $\sigma_i$ which is neither a neighbor of $v$ nor of $v^*$. Let $y$ be the rightmost vertex in $\sigma_i$ with this property. Since $\sigma$ is a LexDFS order there must be a vertex $z$ with $y \prec_\sigma z \prec_\sigma v$ such that $vz \in E(G)$ and $xz \notin E(G)$, due to Lemma~\ref{lemma:3point-lexdfs}. Since $z$ is also adjacent to $v^*$ we can use Lemma~\ref{lemma:3point-lexdfs} again leading to a vertex $u$ with $z \prec_\sigma u \prec_\sigma x$ that is adjacent to $x$ but not to $v^*$. Due to the choice of $y$, vertex $u$ must be to the right of $v$ in $\sigma$. This is contradiction to $p$ being the parent of $x$ in the \cl-tree $T$ of $\sigma$. 

 Since all three steps can be executed in linear time (see Lemma~\ref{lemma:running-time-ordering}), the algorithm has linear running time in total.
 \end{proof}
 
It follows from Observation~\ref{obs:lexdfs+} that Algorithm~\ref{algo:ldfs_ordering} is able to compute any LexDFS order of a chordal graph $G$.
 
 \begin{corollary}
 Algorithm~\ref{algo:ldfs_ordering} can compute any LexDFS order of a chordal graph $G$.
 \end{corollary}
 
This result does not hold for efficient implementations of graph searches in general. One example is \emph{Maximal Neighborhood Search (MNS)} introduced by Corneil and Krueger in 2008~\cite{corneil2008unified} as a generalization of both LexBFS and LexDFS. This search can be implemented with linear running time by implementing LexBFS. However, not every MNS order is a LexBFS order, so this approach can only compute a subset of the MNS orders of a graph.
Observation~\ref{obs:lexdfs+} also leads to an easy recognition algorithm of LexDFS orders.
 
 \begin{corollary}
   LexDFS orders can be recognized in linear time on chordal graphs.
 \end{corollary}

 To illustrate the final procedure of Algorithm~\ref{algo:ldfs_ordering}, we give the following example.
 
 \def\aa{a}
 \def\bb{b}
 \def\cc{c}
 \def\dd{d}
 \def\ee{e}
 \def\ff{f}
 \def\gg{g}
 \def\hh{h}
 \def\ii{i}
 \def\jj{j}

 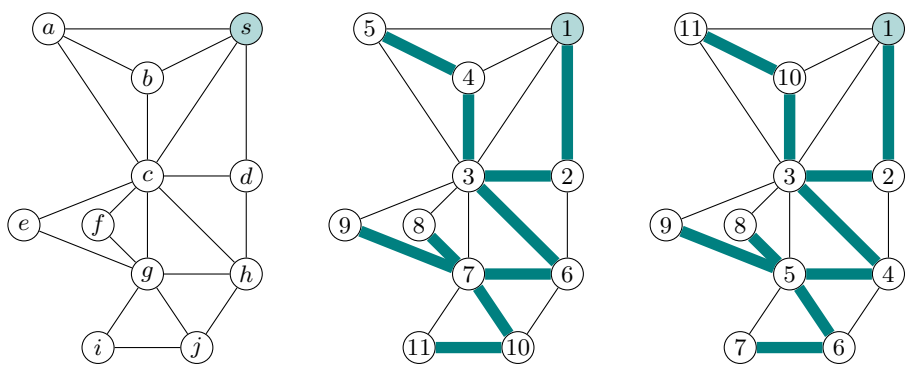
\begin{figure}
 \centering
	\begin{tikzpicture}[scale=.65,
	vertex/.style={inner sep=0pt,minimum size=12pt,draw,circle},
	tree/.style={teal,line width=1.5mm},
	root/.style={inner sep=0pt,minimum size=12pt,draw,circle,fill=teal!30!white}]
	\small
	\begin{scope}[xscale=-1]
	 \node[root] (s) at (0,0) {$s$};
	 \node[vertex] (i) at (4,0) {$\aa$};
	 \node[vertex] (d) at (2,-1) {$\bb$};
	 \node[vertex] (a) at (0,-3) {$\dd$};
	 \node[vertex] (e) at (2,-3) {$\cc$};
	 \node[vertex] (h) at (3,-4) {$\ff$};
	 \node[vertex] (b) at (0,-5) {$\hh$};
	 \node[vertex] (f) at (2,-5) {$\gg$};
	 \node[vertex] (c) at (1,-6.5) {$\jj$};
	 \node[vertex] (g) at (3,-6.5) {$\ii$};
	 \node[vertex] (j) at (4.5,-4) {$\ee$};
	 
	 \draw (s) -- (i);
	 \draw (s) -- (d);
	 \draw (s) -- (a);
	 \draw (s) -- (e);
	 \draw (i) -- (d);
	 \draw (i) -- (e);
	 \draw (d) -- (e);
	 \draw (a) -- (e);
	 \draw (a) -- (b);
	 \draw (e) -- (b);
	 \draw (e) -- (f);
	 \draw (e) -- (h);
	 \draw (b) -- (f);
	 \draw (b) -- (c);
	 \draw (f) -- (h);
	 \draw (f) -- (c);
	 \draw (f) -- (g);
	 \draw (c) -- (g);
	 \draw (j) -- (e);
	 \draw (j) -- (f);
	\end{scope}
	
	\begin{scope}[xshift=6.5cm,xscale=-1]
	 \node[root] (s) at (0,0) {$1$};
	 \node[vertex] (i) at (4,0) {$5$};
	 \node[vertex] (d) at (2,-1) {$4$};
	 \node[vertex] (a) at (0,-3) {$2$};
	 \node[vertex] (e) at (2,-3) {$3$};
	 \node[vertex] (h) at (3,-4) {$8$};
	 \node[vertex] (b) at (0,-5) {$6$};
	 \node[vertex] (f) at (2,-5) {$7$};
	 \node[vertex] (c) at (1,-6.5) {$10$};
	 \node[vertex] (g) at (3,-6.5) {$11$};
	 \node[vertex] (j) at (4.5,-4) {$9$};
	 
	 \draw (s) -- (i);
	 \draw (s) -- (d);
	 \draw[tree] (s) -- (a);
	 \draw (s) -- (e);
	 \draw[tree] (i) -- (d);
	 \draw (i) -- (e);
	 \draw[tree] (d) -- (e);
	 \draw[tree] (a) -- (e);
	 \draw (a) -- (b);
	 \draw[tree] (e) -- (b);
	 \draw (e) -- (f);
	 \draw (e) -- (h);
	 \draw[tree] (b) -- (f);
	 \draw (b) -- (c);
	 \draw[tree] (f) -- (h);
	 \draw[tree] (f) -- (c);
	 \draw (f) -- (g);
	 \draw[tree] (c) -- (g);
	 \draw (j) -- (e);
	 \draw[tree] (j) -- (f);
	\end{scope}

	\begin{scope}[xshift=13cm,xscale=-1]
	 \node[root] (s) at (0,0) {$1$};
	 \node[vertex] (i) at (4,0) {$11$};
	 \node[vertex] (d) at (2,-1) {$10$};
	 \node[vertex] (a) at (0,-3) {$2$};
	 \node[vertex] (e) at (2,-3) {$3$};
	 \node[vertex] (h) at (3,-4) {$8$};
	 \node[vertex] (b) at (0,-5) {$4$};
	 \node[vertex] (f) at (2,-5) {$5$};
	 \node[vertex] (c) at (1,-6.5) {$6$};
	 \node[vertex] (g) at (3,-6.5) {$7$};
	 \node[vertex] (j) at (4.5,-4) {$9$};
	 
	 \draw (s) -- (i);
	 \draw (s) -- (d);
	 \draw[tree] (s) -- (a);
	 \draw (s) -- (e);
	 \draw[tree] (i) -- (d);
	 \draw (i) -- (e);
	 \draw[tree] (d) -- (e);
	 \draw[tree] (a) -- (e);
	 \draw (a) -- (b);
	 \draw[tree] (e) -- (b);
	 \draw (e) -- (f);
	 \draw (e) -- (h);
	 \draw[tree] (b) -- (f);
	 \draw (b) -- (c);
	 \draw[tree] (f) -- (h);
	 \draw[tree] (f) -- (c);
	 \draw (f) -- (g);
	 \draw[tree] (c) -- (g);
	 \draw (j) -- (e);
	 \draw[tree] (j) -- (f);
	\end{scope}
  
 \end{tikzpicture}
 \caption{The graph $G=(V,E)$ on the left side is chordal. 
 In the middle, the \cl-tree of a LexBFS starting at $s$ is shown. Vertex labels correspond to the index in the search order. On the right side, the same tree is shown, but now the labeling fits to a LexDFS starting at $s$.}\label{fig:example}
 \end{figure}

 \begin{example}
  Given the chordal graph in Figure~\ref{fig:example}, start vertex $s$ and $\rho=(a,b,\dots,j,s)$, we begin by computing a LexBFS$^+(\rho)$ order using partition refinement. The first pivot is $s$ with neighborhood $N(s)=\{\dd,\cc,\bb,\aa\}$, which yields the partition $(s)(\dd,\cc,\bb,\aa)(\jj,\ii,\hh,\gg,\ff,\ee)$. With $\dd$ as the next pivot, we obtain $(s)(\dd)(\cc)(\bb,\aa)(\hh)(\jj,\ii,\gg,\ff,\ee)$. After a few more steps, we have $\pi=(s,\dd,\cc,\bb,\aa,\hh,\gg,\ff,\ee,\jj,\ii)$ which is the LexBFS$^+(\rho)$ order. Now, we consider the \cl-tree $T$ induced by this order, which is shown in Figure~\ref{fig:example}.
  
  By Corollary~\ref{corol:ldfs-lbfs-trees}, we see that $T$ is also an \cl-tree of LexDFS rooted in $s$. 
  We first compute the order $\beta$ and, as seen before, we can use any order where the children of a vertex are always to the left of their parent. Thus, we can use $\beta=\pi^-$, although it is not a BFS order of $T$, since $T$ is an \cl-tree and not a standard \cf-tree. 
  
  Now, we iterate through $\beta$ and use partition refinement, beginning with $\mathcal{Q}=\beta$, to compute a final tie-breaking rule $\tau$ (see Algorithm~\ref{algo:ordering}). Vertex $\ii$ has an empty neighborhood to the left so nothing has to be done. Vertex $\jj$ as neighbor $\ii$ to the left in $\beta$, so the first refinement of $\mathcal{Q}$ occurs. After processing vertices $\ee$ and $\ff$ with empty left neighborhoods, we refine for $\gg$ with $\{\ee,\ff,\ii,\jj\}$. This yields the intermediate partition $(\ii)(\jj,\ee,\ff)(\gg,\hh,\aa,\bb,\cc,\dd,s)$. Finally, we obtain $\mathcal{Q}=(\ii)(\jj)(\ee,\ff)(\aa)(\gg)(\hh)(\bb)(\cc)(\dd)(s)$.

  $\mathcal{Q}$ is post-processed to compute $\tau$. Here, we sort $(\ee,\ff)$ with respect to $\rho^-$, which is the only part of $\mathcal{Q}$ with more than one vertex. Furthermore, we move $s$ to the leftmost position and reverse the whole order. This yields $\tau=(\dd,\cc,\bb,\hh,\gg,\aa,\ee,\ff,\jj,\ii,s)$. Now, we perform a final DFS$^+(\tau)$ on the tree $T$. We start in $s$ and follow the unique path to $\cc$. Vertex $\cc$ has two children in $T$ and $\tau$ forces us to take $\hh$, since $\hh$ is to the right of $\bb$ in $\tau$. Similarly, $\jj$ is chosen after $\gg$ due to $\tau$. The final LexDFS$^+(\rho)$ order is $(s,\dd,\cc,\hh,\gg,\jj,\ii,\ff,\ee,\bb,\aa)$. It is shown on the right side of Figure~\ref{fig:example}.
 
 \end{example}

%% file: ldfs-chordal-conclusion.tex
In this paper, we have presented the first linear time implementation of LexDFS on chordal graphs. This is already the second important subclass of perfect graphs, the other being cocomparability graphs, that admits a linear time implementation of LexDFS. In contrast to the algorithm for cocomparability graphs~\cite{kohler2014linear}, however, our approach can compute any LexDFS ordering with arbitrary start vertices. Thus, it also yields the first unrestricted linear time implementation of LexDFS on interval graphs, which are the intersection of cocomparability graphs and chordal graphs. It remains an open question whether this result can be algorithmically exploited to efficiently solve hard problems on other subclasses of chordal graphs besides interval graphs in linear time.

In the light of these results, the question of whether LexDFS can be executed in linear time in general is even more interesting. There are several open questions. Can the search tree approach be extended? Are there other graph classes where rooted \cl-trees of LexBFS and LexDFS coincide or is this a characteristic property of chordal graphs? It is also possible that there are other graph classes and other modifications of graph searches that produce a tree equivalent to an \cl-tree of LexDFS in linear time on graphs of this particular class. 

However, if the answer would be `no', that is, we cannot find a general linear time algorithm for LexDFS, it is an interesting question whether recognizing LexDFS orderings can be done faster than actually generating one. In particular, is it possible to check in linear time whether a given vertex ordering or a search tree belongs to LexDFS?

%% file: ldfs-chordal-main-arxiv.bbl
\begin{thebibliography}{10}

\bibitem{beisegel20??recognition}
Jesse Beisegel, Carolin Denkert, Ekkehard K{\"o}hler, Matja\v{z} Krnc, Nevena
  Piva\v{c}, Robert Scheffler, and Martin Strehler.
\newblock The recognition problem of graph search trees.
\newblock Submitted.

\bibitem{beisegel2019recognizing}
Jesse Beisegel, Carolin Denkert, Ekkehard K{\"o}hler, Matja\v{z} Krnc, Nevena
  Piva\v{c}, Robert Scheffler, and Martin Strehler.
\newblock Recognizing graph search trees.
\newblock In {\em Proceedings of Lagos 2019, the tenth Latin and American
  Algorithms, Graphs and Optimization Symposium}, volume 346 of {\em ENTCS},
  pages 99--110. Elsevier, 2019.
\newblock \href {https://doi.org/10.1016/j.entcs.2019.08.010}
  {\path{doi:10.1016/j.entcs.2019.08.010}}.

\bibitem{berry2009maximal}
Anne Berry, Richard Krueger, and Genevi{\`e}ve Simonet.
\newblock Maximal label search algorithms to compute perfect and minimal
  elimination orderings.
\newblock {\em SIAM Journal on Discrete Mathematics}, 23(1):428--446, 2009.
\newblock \href {https://doi.org/10.1137/070684355}
  {\path{doi:10.1137/070684355}}.

\bibitem{brandstadt1999graph}
Andreas Brandst{\"a}dt, Van~Bang Le, and Jeremy~P. Spinrad.
\newblock {\em Graph Classes: A Survey}.
\newblock SIAM, 1999.
\newblock \href {https://doi.org/10.1137/1.9780898719796}
  {\path{doi:10.1137/1.9780898719796}}.

\bibitem{corneil2013ldfs}
Derek~G. Corneil, Barnaby Dalton, and Michel Habib.
\newblock {LDFS}-based certifying algorithm for the minimum path cover problem
  on cocomparability graphs.
\newblock {\em SIAM Journal on Computing}, 42(3):792--807, 2013.
\newblock \href {https://doi.org/10.1137/11083856X}
  {\path{doi:10.1137/11083856X}}.

\bibitem{corneil2016power}
Derek~G. Corneil, J{\'e}r{\'e}mie Dusart, Michel Habib, and Ekkehard
  K{\"o}hler.
\newblock On the power of graph searching for cocomparability graphs.
\newblock {\em SIAM Journal on Discrete Mathematics}, 30(1):569--591, 2016.
\newblock \href {https://doi.org/10.1137/15M1012396}
  {\path{doi:10.1137/15M1012396}}.

\bibitem{corneil2008unified}
Derek~G. Corneil and Richard~M. Krueger.
\newblock A unified view of graph searching.
\newblock {\em SIAM Journal on Discrete Mathematics}, 22(4):1259--1276, 2008.
\newblock \href {https://doi.org/10.1137/050623498}
  {\path{doi:10.1137/050623498}}.

\bibitem{zbMATH05822744}
Derek~G. {Corneil}, Stephan {Olariu}, and Lorna {Stewart}.
\newblock {The LBFS structure and recognition of interval graphs}.
\newblock {\em {SIAM Journal on Discrete Mathematics}}, 23(4):1905--1953, 2009.
\newblock \href {https://doi.org/10.1137/S0895480100373455}
  {\path{doi:10.1137/S0895480100373455}}.

\bibitem{creusefond2017lexdfs}
Jean Creusefond, Thomas Largillier, and Sylvain Peyronnet.
\newblock A {LexDFS}-based approach on finding compact communities.
\newblock In Mehmet Kaya, {\"O}zcan Erdo{\v g}an, and Jon Rokne, editors, {\em
  From Social Data Mining and Analysis to Prediction and Community Detection},
  pages 141--177. Springer, Cham, 2017.
\newblock \href {https://doi.org/10.1007/978-3-319-51367-6_7}
  {\path{doi:10.1007/978-3-319-51367-6_7}}.

\bibitem{habib2000lex}
Michel Habib, Ross McConnell, Christophe Paul, and Laurent Viennot.
\newblock Lex-{BFS} and partition refinement, with applications to transitive
  orientation, interval graph recognition and consecutive ones testing.
\newblock {\em Theoretical Computer Science}, 234(1-2):59--84, 2000.
\newblock \href {https://doi.org/10.1016/S0304-3975(97)00241-7}
  {\path{doi:10.1016/S0304-3975(97)00241-7}}.

\bibitem{hagerup1985biconnected}
Torben Hagerup.
\newblock Biconnected graph assembly and recognition of {DFS} trees.
\newblock Technical Report A 85/03, Universit\"at des Saarlandes, 1985.
\newblock \href {https://doi.org/10.22028/D291-26437}
  {\path{doi:10.22028/D291-26437}}.

\bibitem{hagerup1985recognition}
Torben Hagerup and Manfred Nowak.
\newblock Recognition of spanning trees defined by graph searches.
\newblock Technical Report A 85/08, Universit\"at des Saarlandes, 1985.

\bibitem{zbMATH03481829}
John {Hopcroft} and Robert {Tarjan}.
\newblock {Efficient planarity testing}.
\newblock {\em {Journal of the ACM}}, 21:549--568, 1974.
\newblock \href {https://doi.org/10.1145/321850.321852}
  {\path{doi:10.1145/321850.321852}}.

\bibitem{kohler2014linear}
Ekkehard K{\"o}hler and Lalla Mouatadid.
\newblock Linear time {LexDFS} on cocomparability graphs.
\newblock In R.~Ravi and Inge~Li G{\o}rtz, editors, {\em Algorithm Theory –
  SWAT 2014}, volume 8503 of {\em LNCS}, pages 319--330, Cham, 2014. Springer.
\newblock \href {https://doi.org/10.1007/978-3-319-08404-6_28}
  {\path{doi:10.1007/978-3-319-08404-6_28}}.

\bibitem{korach1989dfs}
Ephraim Korach and Zvi Ostfeld.
\newblock {DFS} tree construction: {A}lgorithms and characterizations.
\newblock In Jan van Leeuwen, editor, {\em Graph-Theoretic Concepts in Computer
  Science}, volume 344 of {\em LNCS}, pages 87--106, Berlin, Heidelberg, 1989.
  Springer.
\newblock \href {https://doi.org/10.1007/3-540-50728-0_37}
  {\path{doi:10.1007/3-540-50728-0_37}}.

\bibitem{krueger2005graph}
Richard~M. Krueger.
\newblock {\em Graph Searching}.
\newblock PhD thesis, University of Toronto, 2005.
\newblock URL: \url{http://www.cs.toronto.edu/~krueger/papers/thesis.ps}.

\bibitem{ma????}
Tze-Heng Ma.
\newblock Unpublished manuscript.

\bibitem{manber1990}
Udi Manber.
\newblock Recognizing breadth-first search trees in linear time.
\newblock {\em Information Processing Letters}, 34(4):167--171, 1990.
\newblock \href {https://doi.org/10.1016/0020-0190(90)90155-Q}
  {\path{doi:10.1016/0020-0190(90)90155-Q}}.

\bibitem{mertzios2018linear}
George~B. Mertzios, Andr{\'e} Nichterlein, and Rolf Niedermeier.
\newblock A linear-time algorithm for maximum-cardinality matching on
  cocomparability graphs.
\newblock {\em SIAM Journal on Discrete Mathematics}, 32(4):2820--2835, 2018.
\newblock \href {https://doi.org/10.1137/17M1120920}
  {\path{doi:10.1137/17M1120920}}.

\bibitem{rose1976algorithmic}
Donald~J. Rose, R.~Endre Tarjan, and George~S. Lueker.
\newblock Algorithmic aspects of vertex elimination on graphs.
\newblock {\em SIAM Journal on Computing}, 5(2):266--283, 1976.
\newblock \href {https://doi.org/10.1137/0205021} {\path{doi:10.1137/0205021}}.

\bibitem{simon1991new}
Klaus Simon.
\newblock A new simple linear algorithm to recognize interval graphs.
\newblock In Hanspeter Bieri and Hartmut Noltemeier, editors, {\em
  Computational Geometry -- Methods, Algorithms and Applications}, volume 553
  of {\em LNCS}, pages 289--308, Berlin, Heidelberg, 1991. Springer.
\newblock \href {https://doi.org/10.1007/3-540-54891-2_22}
  {\path{doi:10.1007/3-540-54891-2_22}}.

\bibitem{spinrad20??efficient}
Jeremy~P. Spinrad.
\newblock Efficient implementation of lexicographic depth first search.
\newblock Submitted.

\bibitem{tarjan1972depth}
Robert Tarjan.
\newblock Depth-first search and linear graph algorithms.
\newblock {\em SIAM Journal on Computing}, 1(2):146--160, 1972.
\newblock \href {https://doi.org/10.1137/0201010} {\path{doi:10.1137/0201010}}.

\bibitem{xu2013moplex}
Shou-Jun Xu, Xianyue Li, and Ronghua Liang.
\newblock Moplex orderings generated by the {LexDFS} algorithm.
\newblock {\em Discrete Applied Mathematics}, 161(13-14):2189--2195, 2013.
\newblock \href {https://doi.org/10.1016/j.dam.2013.02.028}
  {\path{doi:10.1016/j.dam.2013.02.028}}.

\end{thebibliography}
